\documentclass[acmsmall]{acmart}

\usepackage{amsmath}
\usepackage{amsthm}
\usepackage{mathtools}
\usepackage[inline]{enumitem}
\usepackage{tikz-cd}
\usetikzlibrary{quantikz}
\usepackage{hyperref}
\newcommand{\FOL}{{\mathsf{FOL}}}
\newcommand{\PL}{{\mathsf{PL}}}

\newcommand{\HDPL}{{\mathsf{HDPL}}}
\newcommand{\HDQL}{{\mathsf{HDQL}}}

\newcommand{\Mod}{\mathtt{Mod}}
\newcommand{\Sen}{\mathtt{Sen}}

\newcommand{\Prop}{\mathtt{Prop}}

\newcommand{\W}{\mathcal{W}}
\newcommand{\X}{\mathcal{X}}
\newcommand{\Y}{\mathcal{Y}}
\newcommand{\Hil}{\mathcal{H}}

\newcommand{\N}{\mathbb{N}}

\newcommand{\act}{\mathfrak{a}}
\newcommand{\bact}{\mathfrak{b}}

\newcommand{\cn}{\mathtt{c}}
\newcommand{\hil}{\mathtt{h}}
\newcommand{\vc}{\mathtt{v}}
\newcommand{\ut}{\mathtt{u}}

\newcommand{\pr}{\mathtt{pr}}
\newcommand{\w}{\mathtt{w}}

\newcommand{\red}{\!\upharpoonright\!}

\newcommand{\at}[1]{@_{#1}\,}
\newcommand{\nec}[1]{[#1]}
\newcommand{\pos}[1]{\langle #1 \rangle}
\newcommand{\store}[1]{{\downarrow}#1\,{\cdot}\,}
\newcommand{\Forall}[1]{\forall #1\,{\cdot}\,}
\newcommand{\Exists}[1]{\exists #1\,{\cdot}\,}
\newcommand{\ip}[2]{\langle #1  \mid #2\rangle}
\usepackage{calc}
\usepackage{longtable}

\usepackage{scalerel}

\newcommand{\bbsemicolon}{%
  \scalerel*{%
    \hbox{\usefont{U}{bbold}{m}{n} ;}%
  }{;}%
}
\newcommand{\comp}{\mathbin{\bbsemicolon}}

\usepackage{array,etoolbox}
\newcounter{nrf}

\makeatletter

\newlength{\PS@lastparam}
\newlength{\PSlastparam}
\newcommand{\PSlp}{%
  \setlength{\PSlastparam}{\PS@lastparam}%
  \the\PSlastparam
}
\def\PS@sub@lastparam{}

\newcommand{\PS@numwidth}{99}
\newcommand{\PSnumwidth}[1]{%
  \renewcommand{\PS@numwidth}{#1}%
}

\newcommand{\PS@style}{\small}
\newcommand{\PS@numstyle}{\footnotesize}

\newlength{\PSindent}
\newlength{\PS@extraindent}
\newlength{\PSpre}
\newlength{\PSpost}
\newlength{\PS@Nwidth}
\newlength{\PS@Swidth}
\newlength{\PS@Ewidth}
\newlength{\PScolsep}
\newcommand{\PS@rownumber}{%
  \ifPS@subsubsteps
  \thePSsubstepc.%
  \the\numexpr\value{PSsubsubstepc}+1\relax
  \else
  \ifPS@substeps
  \thePSstepc.%
  \the\numexpr\value{PSsubstepc}+1\relax
  \else
  \the\numexpr\value{PSstepc}+1\relax
  \fi\fi
}
\newcommand{\PS@step}{%
  \ifPS@subsubsteps
  \refstepcounter{PSsubsubstepc}%
  \else
  \ifPS@substeps
  \refstepcounter{PSsubstepc}%
  \else
  \refstepcounter{PSstepc}
  \fi\fi%
}

\newif\ifPS@inprogress
\newif\ifPS@substeps
\newif\ifPS@subsubsteps
\newif\ifPS@continued
\newif\ifPS@subcontinued
\newcounter{PSc}
\newcounter{PSstepc}[PSc]
\newcounter{PSsubstepc}[PSstepc]
\renewcommand{\thePSsubstepc}{\thePSstepc.\arabic{PSsubstepc}}
\newcounter{PSsubsubstepc}[PSsubstepc]

\newenvironment{proofsteps}[1]{%
  \global\settowidth{\PS@lastparam}{\PS@style\hspace*{#1}}
  \ifPS@continued\else\refstepcounter{PSc}\fi
  \begingroup
  \setlength{\LTpre}{\PSpre}%
  \setlength{\LTpost}{\PSpost}%
  
  \setlength{\tabcolsep}{0pt}
  \noindent\PS@style
  \settowidth{\PS@Nwidth}{\PS@numstyle\PS@numwidth}%
  \setlength{\PS@Swidth}{#1}%
  \addtolength{\PS@Swidth}{-\PS@extraindent}%
  \setlength{\PS@Ewidth}{\linewidth}%
  \addtolength{\PS@Ewidth}{-\PSindent}%
  \addtolength{\PS@Ewidth}{-\PS@extraindent}%
  \addtolength{\PS@Ewidth}{-\PS@Nwidth}%
  \addtolength{\PS@Ewidth}{-\PScolsep}%
  \addtolength{\PS@Ewidth}{-\PS@Swidth}%
  \addtolength{\PS@Ewidth}{-\PScolsep}%
  \PS@inprogresstrue
  \longtable{%
    @{\hspace*{\PSindent}\hspace*{\PS@extraindent}\makebox[\PS@Nwidth][r]{\PS@rownumber}}%
    @{\hskip\PScolsep}>{\PS@step}p{\PS@Swidth}%
    @{\hskip\PScolsep}>{\footnotesize\raggedright\arraybackslash}p{\PS@Ewidth}%
  }%
}{%
  \ifPS@inprogress
  \addtocounter{table}{-1}%
  \endlongtable  
  \endgroup
  \PS@continuedfalse
  \PS@inprogressfalse
  \else\fi
}

\newcommand{\PSbreak}[1]{%
  \endproofsteps
  \par\medskip
  #1
  \medskip\par
  \PS@continuedtrue
  \proofsteps{\PS@lastparam}%
}

\newif\ifPS@sub@inprogress

\newif\ifPS@laststep
\newcommand{\laststep}{\global\PS@laststeptrue}
\newif\ifPS@lastsubstep
\newcommand{\lastsubstep}{\global\PS@lastsubsteptrue}

\newcommand{\adjustcol}[1]{%
  \global\advance\@colroom-#1%
}

\makeatother
\newcommand{\psqed}{%
  \vspace{-\baselineskip}\vspace{-1\smallskipamount}
}

\newcommand*{\pcformat}[1]{%
  [\;{\normalfont\itshape #1}\;]%
}

\newenvironment{proofcases}[1][]{%
  \description[font=\pcformat, leftmargin=\parindent, #1]%
}{\enddescription}

\AtBeginDocument{%
  }

\acmJournal{ARXIV}

\begin{document}

\title{Foundations of logic programming in hybrid-dynamic quantum logic}

\author{Daniel G\u{a}in\u{a}}
\email{daniel@imi.kyushu-u.ac.jp}
\orcid{0000-0002-0978-2200}
\affiliation{%
  \institution{Kyushu University}
  \streetaddress{744 Motooka Nishi-ku}
  \city{Fukuoka}
  \country{Japan}
  \postcode{819-0395}
}
\renewcommand{\shortauthors}{D. G\u{a}in\u{a}}

\begin{abstract}
The main contribution of the present paper is the introduction of a simple yet expressive hybrid-dynamic logic for describing quantum programs.  
This version of quantum logic can express quantum measurements and unitary evolutions of states in a natural way based on concepts advanced in (hybrid and dynamic) modal logics.
We then study Horn clauses in the hybrid-dynamic quantum logic proposed, and develop a series of results that lead to an initial semantics theorem for sets of clauses that are satisfiable.
This shows that a significant fragment of hybrid-dynamic quantum logic has good computational properties, and can serve as a basis for defining executable languages for specifying quantum programs.
We set the foundations of logic programming in this fragment by proving a variant of Herbrand’s theorem, which reduces the semantic entailment of a logic-programming query by a program to the search of a suitable substitution.
\end{abstract}



\maketitle

\section{Introduction}\label{sec:intro}
In the present paper, we propose a hybrid-dynamic quantum logic as foundation for quantum computing.
The present version of quantum logic is a modal logic enhanced with features from both hybrid and dynamic logics. 
Dynamic propositional logic  is suitable for describing and reasoning about classical computer programs, which means that
it is natural to develop a quantum version of dynamic propositional logic thus providing a basis for the formal verification of quantum programs~\cite{DBLP:journals/mscs/BaltagS06}.
On the other hand, the ability of naming individual states of hybrid logics has several advantages.
For example, hybrid logics allow a more uniform proof theory and model theory than non-hybrid modal logics \cite{brau11}.
The additional infrastructure geared towards distinguishing states and reasoning about their properties is crucial in applications to formal methods.
One can express temporal properties using the sentence building operators \emph{retrieve} and \emph{store} from hybrid logics~\cite{DBLP:journals/jsyml/ArecesBM01}.
As a result, hybrid logics were recognized as suitable for providing logical foundation for the reconfiguration paradigm~\cite{Madeira13,DBLP:journals/scp/MadeiraNBM16,DBLP:journals/jacm/Gaina20}.
There are at least two main challenges for designing a quantum logic in the modal logic tradition:
\begin{enumerate}
\item The set of states of a quantum program is given by a Hilbert space, which is a two sorted first-order model, since its elements are divided into scalars and vectors. 
The algebraic structure of Hilbert spaces, in general, and the equations on both scalars and vectors, in particular, play an important role in formal specification and verification of quantum systems. 
This suggests a two-layered approach to the design and analysis of quantum programs, involving
\emph{a local view}, which amounts to describing the structural algebraic  properties of quantum states, and
\emph{a global view}, which corresponds to a specialized language for reasoning about the way the quantum systems evolve.
\item Hilbert spaces are vectorial spaces equipped with an \emph{inner product}, in which each Cauchy sequence of vectors has a limit. Obviously, this is not a first-order property which makes the logical developments significantly more difficult.
We employ the method of diagrams proposed by Robinson in classical model theory to define Hilbert spaces.
We use constant symbols  to stand for the elements of the Hilbert space to be constructed, and we work within the theory which contains all the equations and relations satisfied by that Hilbert space.
In other words, the signatures of nominals are not single-sorted first-order signatures consisting only of constants and relations, but positive diagrams of concrete Hilbert spaces. 
In this way, we describe the structural properties of quantum states mentioned in the first item above.
\end{enumerate}
System dynamics  is modeled by labelled transition systems, that is, Kripke structures, in which \emph{actions} are represented as binary relations between states.
Both \emph{projective measurements} and \emph{unitary evolutions} are regarded as basic modal actions.
The unitary transformations on Hilbert spaces are also known as \emph{quantum gates}.
As argued in \cite{DBLP:journals/mscs/BaltagS06}, this approach allows one to express in a modal logic (based on a local Boolean satisfaction) quantum properties captured traditionally by (non-Boolean) Quantum Logic.
For example, the \emph{orthocomplement}, also called \emph{quantum negation},  $\sim\varphi$, is defined as the set of all vectors orthogonal on the vectors where $\varphi$ holds, while \emph{quantum disjunction} $\varphi_1\oplus \varphi_2$ expressing superpositions is  $\sim(\sim\varphi_1\wedge \sim\varphi_2)$.

In the logical setting described  above, we propose a notion of \emph{Horn clause} which allows one the use of orthocomplement, both quantum and classical implication, retrieve, store, universal quantifier and necessity over structural actions which, in turn, are constructed from \emph{projective measurements} and \emph{unitary transformations}.
Due to quantum negation, not all sets of Horn clauses are satisfiable.
Despite this fact the existence of initial models is guaranteed if there are no inconsistencies.
Therefore, the present results lead to a deeper understanding of initiality which is no longer connected to Horn clauses that are always satisfiable.

Initial semantics has received a lot of attention in formal methods, where it has been linked to formalisms that support execution by means of rewriting.
In algebraic specification, initiality supports the execution of specification languages by rewriting which, in turn, increases the automation of the formal verification process.
In logic programing, the initial models are called the least Herbrand models, and the initiality supports the execution of programs by resolution.
Our approach to initiality is layered and it is directly connected to the structure of sentences, in the style of \cite{DBLP:journals/logcom/GainaF15} developed for first-order logics, and \cite{DBLP:journals/tcs/Gaina17} developed for hybrid logics.
Firstly, we prove the existence of initial models for some \emph{basic sets of sentences} obtained from propositional symbols by applying conjunction, retrieve, store and necessity.
Then the initiality property is shown to be closed under the sentence building operators used to construct Horn clauses: 
both quantum and classical implications, 
orthocomplement, 
retrieve,
store,
universal quantification and 
necessity over actions constructed from both unitary transformations and projective measurements.
In this paper, we explore the role of initiality in logic programming and we prove a version of Herbrand's theorem.
This result establishes an equivalence between the semantic entailment of a query and the existence of an answer substitution for that query --- which is a syntactic construct, and can subsequently be obtained by proof-theoretic means.

The conceptual paper \cite{Baltag2011-BALQLA} advocates a new view on quantum logic, seen as a dynamic logic of some \emph{non-classical information flow}.
It is argued that the correct interpretation of quantum-logical connectives is in terms of dynamic modalities, rather than purely propositional operators.
This philosophical interpretation is supported by a series of technical developments reported in \cite{BaltagS2005,DBLP:journals/mscs/BaltagS06,DBLP:journals/sLogica/BaltagS08}.
The present approach is based on the same ideas, but it departs fundamentally from any of those studies due to the fact that the set of states is not the set of one-dimensional closed linear subspaces of some Hilbert but the entire set of vectors of a Hilbert space.
Our choice is motivated by practical applications to formal methods where the specification of quantum systems is easier in a logical setting which includes all the elements of the underlying Hilbert space.
The dynamic quantum logic proposed in \cite{DBLP:journals/mscs/BaltagS06} is equipped with features to capture properties of composite quantum systems.
We are planning to use the notion of signature morphism and define appropriate \emph{specification building operators}~\cite{DBLP:journals/iandc/SannellaT88} for constructing complex quantum systems from simple and easy-to-verify ones according to the ideas advanced in algebraic specification~\cite{DBLP:series/wsscs/SannellaT93}.

\section{Hilbert spaces} \label{sec:hilbert}
The present section contains a short presentation of Hilbert spaces and some of their properties from a first-order logic ($\FOL$) perspective.
The signature of Hilbert spaces $\Sigma^\hil=(S^\hil,F^\hil,P^\hil)$ is a first-order signature consisting of:
\begin{enumerate}
\item a set of sorts $S^\hil=\{\cn,\vc\}$, where $\cn$ stands for the sort of complex numbers and 
$\vc$ for the sorts of vectors.

\item a set of function symbols $F^\hil=F^\cn\cup F^\vc$, where 
\begin{enumerate}
\item $F^\cn$ consists of the usual operations on complex numbers such as 
addition $\_+\_:\cn~\cn\to\cn$ and its identity element $0:\to\cn$,
multiplication $\_*\_:\cn\to\cn$, etc;
\item $F^\vc=\{ 
\_+\_:\vc~\vc\to\vc, 
0:\to \vc, 
\_\_ : \cn~\vc\to \vc, 
\ip{\_}{\_}:\vc~\vc\to\cn \}$, where 
$\_+\_:\vc~\vc\to\vc$ stands for the vector addition,
$0$ stands for the origin vector, 
the juxtaposition $\_\_ : \cn~\vc\to \vc$ stands for the scalar multiplication, and
$\ip{\_}{\_}:\vc~\vc\to\cn$ stands for the inner product;
\end{enumerate}
\item $P=\{\_<\_:\cn~\cn\}$, where $\_<\_:\cn~\cn$ stands for the ordering among real numbers.
\end{enumerate}

A Hilbert space is a first-order model $\Hil$ over the signature $\Sigma^\hil$ such that:
\begin{enumerate}
\item $\Hil_\cn=\mathbb{C}$, the set of complex numbers;

\item $\Hil$ interprets all function symbols in $F^\cn$ as the usual operations on complex numbers;

\item $\Hil_\vc$ is a set of vectors,
$+^\Hil:\Hil_\vc\times \Hil_\vc\to \Hil_\vc$ is the vector addition, 
${\_\_}^\Hil: \Hil_\cn\times \Hil_\vc\to \Hil_\vc$ is the scalar multiplication,
$\ip{\_}{\_}^\Hil:\Hil_\vc\times \Hil_\vc\to \Hil_\cn$ is the inner product in which each Cauchy sequence of vectors has a limit; 
\footnote{Recall that a sequence of vectors $\{w_n\}_{n\in \N}$ is a Cauchy sequence if for any $\varepsilon >0$  there exists $n\in\N$ such that $||w_i-w_j||<\varepsilon$ for all $i,j\geq n$, where $||w||=\sqrt{\ip{w}{w}}$ is the length/norm of a vector $w\in \Hil_\vc$.}
\item $<^\Hil$ is the usual total (strict) ordering on $\mathbb{R}$.
\end{enumerate}

Two vectors $w_1,w_2\in \Hil_\vc$ are orthogonal, in symbols, $w_1\perp w_2$ if $\ip{w_1}{w_2}^\Hil=0^\Hil$.
This notational convention can be extended to sets of vectors in the usual way.
A closed subspace $\X$ of a Hilbert space $\Hil$ is a (first-order) substructure $\X\subseteq \Hil$ such that
\begin{enumerate*}[label=(\alph*)]
\item $\X_\cn=\mathbb{C}$ and 
\item each Cauchy sequence of vectors  from $\X$ has a limit in $\X$. 
\end{enumerate*}
In other words, a closed subspace is in particular a Hilbert space.
A closed subspace is usually identified by its set of vectors --- 
a convention which is adopted in this paper.
The orthocomplement of $\X$ denoted $\X^\perp$ is defined by $\X^\perp=\{w\in \Hil_\vc \mid w\perp \X \}$.
The orthocomplement of a set of vectors is a closed subspace. 
The direct sum of two closed subspaces $\X$ and $\Y$ is defined as $\X\oplus \Y\coloneqq \{x +^\Hil y \mid x\in \X \text{ and } y\in\Y \}$.
The following results are well-known.
\begin{theorem} \label{th:hilbert}
Let $\Hil$ be a Hilbert space.
\begin{enumerate}
\item \label{th:hilbert-1} A substructure $\X\subseteq \Hil$ is a closed subspace iff $\X^{\perp\perp}=\X$.
\item \label{th:hilbert-2} $\X\oplus \Y= (\X^\perp \cap \Y^\perp)^\perp$ for all closed subspaces $\X\subseteq \Hil$ and $\Y\subseteq \Hil$.
\item \label{th:hilbert-3} $\X^{\perp\perp}=\Y\cap (\Y\cap \X^\perp)^\perp$ for all closed subspaces  $\X\subseteq \Y\subseteq \Hil$, which means that the (global) closure of $\X$ is equal to local closure of $\X$ relative to $\Y$.
\item $\Hil=\X\oplus \X^\perp$ for all closed subspaces $\X\subseteq \Hil$.
\end{enumerate}
\end{theorem}
By the theorem above, for any closed subspace $\X\subseteq \Hil$ and vector $w\in \Hil$, we have $w=w_1 +^\Hil w_2$, where $w_1\in \X$ and $w_2\in \X^\perp$.
Since $\X\cap \X^\perp=\{0^\Hil\}$, the vectors $w_1\in \X$ and $w_2\in \X^\perp$ are uniquely determined.
The projection $P_\X:\Hil\to \X$ is defined by $P_\X(w)=w_1$, where $w=w_1+w_2$, $w_1\in \X$ and $w_2\in \X^\perp$.
\begin{theorem} \label{th:projection}
For all closed subspaces $X$ and $Y$ of a Hilbert space $\Hil$, we have
\begin{enumerate}
\item $P_\X(\Y)= \X\cap(\X^\perp\oplus \Y)$, and
\item $P_\X^{-1}(\Y) =\X^\perp\oplus(\X\cap \Y)$.
\end{enumerate}
\end{theorem}
We denote $\X^\perp\oplus(\X\cap \Y)$ by $\X\leadsto \Y$. 
This defines a binary operation on closed subspaces called ``Sasaki hook''.
\section{Hybrid-Dynamic Quantum Logic}
We define Hybrid Dynamic Quantum Logic ($\HDQL$)    as an extension of Hybrid-Dynamic Propositional Logic ($\HDPL$) with some constraints on the possible world semantics.

\subsection{Signatures}
The signatures of $\HDQL$ are of the form $\Delta=(\Sigma,E,\Prop)$, where:
\begin{enumerate}
\item $\Sigma$ is a first-order signature obtained from the signature of Hilbert spaces $\Sigma^\hil=(S^\hil,F^\hil,P^\hil)$ by adding 
\begin{enumerate*}[label=(\alph*)]
\item a set $U$ of unitary transformation symbols, which are functions symbols of the form $\ut:\vc\to\vc$;
\item a set $Q$ of quantum measurement symbols, which are function symbols of the form $q:\vc\to\vc$;
\item a set $D$ of constants of sort $\vc$; and 
\item a set $C$ of constants of sort $\cn$;
\end{enumerate*}

\item $E$ is a set of first-order sentences over $\Sigma$;

\item $\Prop$ is a set of propositional symbols which contains a subset $\Prop^c$ of closed propositional symbols.
\end{enumerate}
Notice that the signature of Hilbert spaces is contained in every $\HDQL$ signature.
In applications, the set $E$ of first-order sentences over the  signature $\Sigma$  will play an important role in defining logical properties of quantum models. 
Therefore, we decided to introduce them at the level of $\HDQL$ signatures, since in applications, the set of first-order sentences that describes Hilbert spaces is defined before the specification of a quantum protocol.
A concrete example of $E$ is the positive Robinson diagram of some Hilbert space, that is, all the equations and relations satisfied by some Hilbert space.
We let $\Delta$  range over signatures of the form $(\Sigma,E,\Prop)$ as described above.
Similarly, for any index $i$, we let $\Delta_i$  range over signatures of the form $(\Sigma_i,E_i,\Prop_i)$, where $\Sigma_i=(S^\hil,F^\hil\cup U_i \cup Q_i \cup  D_i \cup C_i,P^\hil)$.
Signature morphisms $\chi:\Delta_1\to \Delta_2$ consists of:
\begin{enumerate}
\item a first-order theory morphism $\chi:(\Sigma_1,E_1)\to (\Sigma_2,E_2)$, which is the identity on $\Sigma^\hil$ and preserves:
\begin{enumerate*}[label=(\alph*)]
\item unitary transformations, $\chi(U_1)\subseteq U_2$,
\item projective measurements, $\chi(Q_1)\subseteq Q_2$, 
\item vector constants, $\chi(D_1)\subseteq D_2$, and
\item scalar constants, $\chi(C_1)\subseteq C_2$;
\end{enumerate*}
\footnote{Since $\chi:(\Sigma_1,E_1)\to(\Sigma_2,E_2)$ is a first-order theory morphism, $E_2$ satisfies $\chi(E_1)$, in symbols, $E_2\models^\FOL\chi(E_1)$.}
and
\item a mapping $\chi:\Prop_1\to\Prop_2$ on propositional symbols such that both closed and non-closed  propositional symbols are preserved, that is, $\chi(\Prop_1^c)\subseteq \Prop_2^c$ and $\chi(\Prop_1\setminus \Prop_1^c) \subseteq \Prop_2\setminus \Prop_2^c$.
\end{enumerate}
We overloaded the notation such that $\chi$ denotes not only the signature morphism from $\Delta_1$ to $\Delta_2$ but also its restrictions to $\Sigma_1$ and $\Prop_1$.
\subsection{Models}
A \emph{quantum model} over a signature $\Delta$ is a Kripke structure $(W,M)$ such that:
\begin{enumerate}
\item $W$ is a first-order structure for the theory $(\Sigma,E)$ such that 
\begin{enumerate}
\item $W\red_{\Sigma^\hil}$ is a Hilbert space,
\item for all $\ut:\vc\to\vc\in U$ the function $\ut^W:W_\vc\to W_\vc$ is a unitary transformation, that is, 
$\ut^W$ is a bounded linear operation which has an adjoint $(\ut^\dagger)^W$ that is its inverse, $\ut^W ; (\ut^\dagger)^W = (\ut^\dagger)^W ; \ut^W = 1_\W$, and
\item for all $q:\vc\to\vc\in Q$ the function $q^W:W_\vc\to W_\vc$ is a quantum measurement, that is, there exists a closed subspace $\X\subseteq W_\vc$ such that $q^W(w)= P_\X(w)/\sqrt{\pr_\X(w)}$, 
where 
\begin{enumerate*}
\item $P_\X:W_\vc\to \X$ is the projection on $\X$, and
\item $\pr_\X: W_\vc\to W_\cn$ is the probability function defined by $\pr_\X(w)=\ip{w}{P_\X(w)}$ for all $w\in W_\vc$.
\end{enumerate*}
\end{enumerate}

\item $M:W_\vc\to |\Mod^\PL(\Prop)|$ is a mapping from the set of vectors $W_\vc$ to the class of propositional models over $\Prop$ such that $r^{(W,M)}\coloneqq \{w\in W_\vc \mid r\in M_w\}$ is a closed subspace for all closed propositional symbols $r\in \Prop^c$.
\footnote{Notice that $M_w$ is the propositional model at the state $w$, in other words, $M$ applied to $w$.}
\end{enumerate}
A $\HDQL$ homomorphism $h:(W,M)\to(V,N)$ is a first-order homomorphism $h:W\to V$ such that 
\begin{enumerate*}[label=(\alph*)]
\item $h_\cn:W_\cn\to V_\cn$ is the identity, and 
\item $M_w\subseteq N_{h(w)}$ for all vectors $w\in V_\vc$.
\end{enumerate*}
\begin{lemma} \label{lemma:inj}
All homomorphisms of Hilbert spaces are injective. 
In particular, all homomorphisms of quantum models are injective.
\end{lemma}
\begin{proof}
Let $h:(W,M)\to(V,N)$ be a homomorphism of quantum models.
For all vectors $w\in W_\vc$, we have
$h(w)=0^V$ iff (by inner product definition)
$\ip{h(w)}{h(w)}^V= 0$ (since $h$ is a first-order homomorphism) iff 
$h(\ip{w}{w}^W)=0$ iff ($h$ is the identity on complex numbers) 
$\ip{w}{w}^W=0$ iff (by the definition of the inner product)
$w= 0^W$.
Then for all vectors $w_1,w_2\in W_\vc$, 
we have  
$h(w_1)=h(w_2)$ iff 
$h(w_1) -^V h(w_2)=0^V$ iff 
$h(w_1-^W w_2)=0^V$ iff
$w_1-^W w_2=0^W$ iff
$w_1=w_2$.
Hence, $h$ is injective.
\end{proof}

The class of quantum models over a signature $\Delta$ together with their homomorphisms forms a category denoted $\Mod(\Delta)$.

For any set of vector variables $X$ for a signature $\Delta=(\Sigma,E,\Prop)$, 
we denote by $\Delta[X]$ the signature $(\Sigma[X],E,\Prop)$, where $\Sigma[X]$ is the first-order signature obtained from $\Sigma$ by adding all variables in $X$ as new constants.
For any quantum model $(W,M)$ over $\Delta$, 
by a $\Delta[X]$-\emph{expansion} of $(W, M)$ we understand a quantum model $(W',M)$ over $\Delta[X]$ that interprets all symbols in $\Delta$ in the same way as $(W, M)$. 
In other words, $(W,M)$ is a \emph{reduct} of $(W',M)$ to the signature $\Delta$, in symbols, $(W',M)\red_\Delta=(W,M)$.
This means that $W$ is the reduct of $W'$ to the first-order signature $\Sigma$, in symbols, $W'\red_\Sigma=W$.  
 
\subsection{Sentences}
The set of \emph{actions} over a signature $\Delta$ is defined by the following grammar:
\begin{center}
$\act\Coloneqq \ut \mid q \mid \act \comp \act \mid \act \cup \act \mid \act^{*}$
\end{center}
where $\ut:\vc\to\vc\in U$ is a unitary transformation symbol, and $q:\vc\to\vc\in Q$ is a projective measurement symbol.
The set of \emph{sentences} over $\Delta$, denoted $\Sen(\Delta)$, is defined by the following grammar: 
\begin{center}
$\gamma \Coloneqq
p \mid
\at{k}\gamma \mid
\gamma \wedge \gamma \mid
\neg\gamma\mid
\sim \gamma \mid
\nec{\act}\gamma \mid
\store{z}\gamma_1\mid
\Forall{X}\gamma_2$
\end{center}
where $p\in\Prop$ is any propositional symbol,
$k\in T_{\Sigma,\vc}$ is a term of sort $\vc$,
$\act$ is an action,
$z$ is a vector variable for $\Delta$,
$\gamma_1$ is a sentence over $\Delta[z]$,
$X$ is a set of vector variables for $\Delta$, and
$\gamma_2$ is a sentence over $\Delta[X]$.
We refer to the sentence-building operators, in order, as 
\emph{retrieve},
\emph{conjunction},
\emph{negation},
\emph{quantum negation},
\emph{necessity}, 
\emph{store}, and
\emph{universal quantification},
respectively.

Each signature morphism $\chi:\Delta_1\to \Delta_1$ induces a sentence translation $\chi:\Sen(\Delta_1) \to\Sen(\Delta_2)$ that replaces, in an inductive manner, in any sentence $\gamma\in\Sen(\Delta_1)$ the symbols from $\Delta_1$ with symbols from $\Delta_2$ according to $\chi:\Delta_1\to\Delta_2$. 

\subsection{Local satisfaction relation}
Let $(W,M)$ be a quantum model over a signature $\Delta$.

The semantics of actions is defined as follows:
\begin{itemize}
\item $(\act_1\comp\act_2)^W=\act_1^W\comp \act_2^W$;
\item $(\act_1\cup\act_2)^W=\act_1^W\cup \act_2^W$;
\item $(\act^*)^W=\bigcup_{n\in \N} (\act^n)^W$,
where $\act^0$ denotes the identity, and $\act^{n+1}=\act\comp \act^n$ for all $n\in \N$.
\end{itemize}

The semantics of sentences is defined as follows:
\begin{itemize}
\item $p^{(W,M)}=\{w\in W_\vc\mid p\in M_w\}$ for all propositional symbols $p$ in $\Delta$;
\item $(\at{k}\gamma)^{(W,M)}=
\left\{\begin{array}{l l}
W_\vc & \text{if } k^W\in \gamma^{(W,M)}, \\
\emptyset  & \text{if } k^W\not\in \gamma^{(W,M)};
\end{array}\right.$
\item $(\gamma_1\wedge\gamma_2)^{(W,M)}= \gamma_1^{(W,M)} \cap \gamma_2^{(W,M)}$;
\item $(\neg\gamma)^{(W,M)}=W_\vc \setminus \gamma^{(W,M)}$;
\item $(\sim\gamma)^{(W,M)}= (\gamma^{(W,M)})^\perp$;
\item $(\nec{\act}\gamma)^{(W,M)}= \{w\in W_\vc\mid \act^W(w)\subseteq \gamma^{(W,M)}\}$,

where $\act^W(w)$ denotes the set of vectors $\{v\in W_\vc\mid (w,v)\in \act^W\}$;
\item $(\store{z}\gamma_1)^{(W,M)}=\{w\in W_\vc \mid w\in \gamma_1^{(W^{z\leftarrow w},M)} \}$,

where $(W^{z\leftarrow w},M)$ is the unique $\Delta[z]$-expansion of $(W,M)$ which interprets $z$ as $w$;
\item $(\Forall{X}\gamma_2)^{(W,M)}=\displaystyle\bigcap_{W_2\red_\Sigma=W} \gamma_2^{(W_2,M)}$.
\end{itemize}

We say that $(W,M)$ \emph{satisfies $\gamma$ in the state $w$}, in symbols, $(W,M)\models^w \gamma$, if $w\in \gamma^{(W,M)}$.

The logic defined above provides considerable expressive power over quantum models.
For example, the \emph{Until} operator can be defined using the hybrid logic operators retrieve and store:
$$Until(\gamma,\psi)=\store{x}\pos{\act} \store{y} \gamma \wedge \at{x} (\nec{\act}(\pos{\act} y \Rightarrow \psi))$$
The current state is named $x$ and 
then $\pos{\act}$ is used to move to an accessible state, which is named $y$.
The first argument of conjunction is $\gamma$, which must hold in the state denoted by $y$.
The second argument of conjunction sets the current state to $x$ by applying $\at{x}$ and then it says that $\psi$ holds in any state which succeeds  $x$ and preceds $y$.
Therefore, one can express temporal properties in the logical framework proposed in this paper.

Another approach to quantum logic that integrates both dynamic and temporal aspects of quantum systems is \cite{takagi2023};
the unitary transformations are regarded  as temporal operators, but the until operator is not discussed for the sake of simplicity. 
The author also claims that none of the existing approaches can deal with both dynamic and temporal aspects of quantum systems.
In this paper, we have provided an alternative in which one can express temporal properties with ease.

\subsection{Global satisfaction relation}
The global satisfaction relation between models and sentences over a signature $\Delta$ is defined as follows:
\begin{itemize}
\item A model $(W,M)$ globally satisfies a sentence $\gamma$, in symbols, $(W,M)\models \gamma$, if $\gamma^{(W,M)}=W_\vc$.
\end{itemize}
In formal methods, the global satisfaction relation is at the core of formal verification, 
since the engineers need to model software and hardware systems with sets of sentences that need to be satisfied globally.
Therefore, the global satisfaction relation between models and sentences is extended to a satisfaction relation between sets of sentences.
\begin{itemize}
\item A set of sentences $\Gamma$ satisfies a sentence $\gamma$, in symbols, $\Gamma\models\gamma$, if $(W,M)\models\Gamma$ implies $(W,M)\models\gamma$ for all quantum models $(W,M)$.
\end{itemize}
The satisfaction relation between sentences defined in this paper is based on global satisfaction between models and sentences, which is different from the standard one used in modal logic literature;
notice that $\emptyset \models\bigwedge\Gamma\Rightarrow \gamma$ implies $\Gamma\models\gamma$ but the backwards implication does not hold.
\begin{itemize}
\item We write $\Gamma\models^k\gamma$ if $\Gamma\models\at{k}\gamma$ holds.
\end{itemize}
Notice that $\emptyset \models^k \bigwedge\Gamma\Rightarrow \gamma$ implies $\Gamma\models^k\gamma$ but the backwards implication does not hold.

\subsection{Substitutions}
The concept of substitution in $\HDQL$ is inherited  from first-order logic, where one can infer new sentences by substituting terms for variables.
Given a signature $\Delta=(\Sigma,E,\Prop)$ and two sets of vector variables $X$ and $Y$ for $\Delta$, 
any mapping $\theta:X\to T_\Sigma(Y)_\vc$ of variables to terms of sort vector determines:
\begin{enumerate}
\item a function $\theta:\Sen(\Delta[X])\to\Sen(\Delta[Y])$ which replaces, in an inductive manner, in each sentence $\gamma\in\Sen(\Delta[X])$ any variable from $x\in X$ with the term $\theta(x)$;
\item a functor $\red_\theta:\Mod(\Delta[Y])\to\Mod(\Delta[X])$ which maps 
 any quantum model $(W,M)$ over $\Delta[Y]$ to a quantum model $(V,M)$ such that
\begin{enumerate*}[label=(\alph*)]
\item $V$ interprets all symbols in $\Delta$ in the same way that $W$ interprets them, and 
\item $x^V=\theta(x)^W$ for all variables $x\in X$.
\end{enumerate*}
\end{enumerate}
\begin{lemma}[Satisfaction condition for substitutions] \label{lemma:subst}
Given a substitution $\theta:X\to T_\Sigma(Y)_\vc$, we have 
$(W,M)\red_\theta\models^w\gamma$ iff $(W,M)\models^w\theta(\gamma)$,
for all sentences $\gamma$ over $\Delta[X]$, 
all quantum models $(W,M)$ over $\Delta[Y]$, and 
all vectors $w\in W_\vc$.
\end{lemma}
The proof of Lemma~\ref{lemma:subst} is conceptually identical to the proof of \cite[Corollary 39]{DBLP:journals/tcs/Gaina17}.

\subsection{Closed sentences}
In this section, we study sentences whose semantics consists of a closed subspace.
For example, $r^{(W,M)}$ is a closed subspace, for all closed propositional symbols $r$ and all quantum models $(W,M)$ over the same signature.
Certain sentence operators such as quantum negation and conjunction preserve this property generating a distinguished class of \emph{closed sentences}.

The set of unitary actions is defined by the following grammar,
\begin{center}
$\bact\Coloneqq 
u\mid
\bact\comp\bact\mid
\bact\cup\bact \mid
\bact^*$
\end{center}
where $u:\vc\to\vc\in U$ is a unitary transformation symbol.
The set of closed sentences, denoted $\Sen_c(\Delta)$, is defined by the following grammar:
\begin{center}
$\rho\Coloneqq 
r\mid
\sim \rho \mid 
\rho \wedge\rho \mid 
\nec{\bact}\rho\mid
\Forall{X}\rho_1$
\end{center}
where 
$r\in\Prop^c$ is a closed propositional symbol,
$\bact$ is a unitary action,
$X$ is a set of vector variables for $\Delta$, and
$\rho_1$ is a closed sentence over $\Delta[X]$.
\begin{theorem}\label{th:Hibert-sat}
The semantics of any closed sentence is a closed subspace, i.e.,
given a signature $\Delta$, 
for all models $(W,M)\in|\Mod(\Delta)|$ and 
all closed sentences $\rho\in\Sen_c(\Delta)$, 
$\rho^{(W,M)}$ is a closed subspace.
\end{theorem}
\begin{proof}
Assume a quantum model $(W,M)$ and a closed sentence $\rho$ over a signature $\Delta$.
We show that $\rho^{(W,M)}$ is a closed subspace.
We proceed by induction on the structure of $\rho$. 
\begin{proofcases}
\item[$r\in\Prop^c$] By the definition of quantum models.
\item [$\sim\rho$] 
Straightforward, since orthogonal complements are closed subspaces.
\item [$\rho_1\wedge\rho_2$]
Straightforward, since the intersection of closed spaces is a closed subspace.
\item [$\nec{\bact}\rho$]
We prove that $\nec{\bact}\X=\{w\in W_\vc\mid \bact^W(w)\subseteq X\}$ is a closed subspace for all closed subspaces $\X\subseteq W_\vc$.
We proceed by induction on the structure of the linear action $\bact$.
\begin{proofcases}[itemsep=1ex]
\item[$\ut\in U$] Straightforward, since unitary transformations are isomorphisms.
\item [$\bact_1\comp \bact_2$] 
By the induction hypothesis applied to $\bact_2$, $\nec{\bact_2}\X$ is a closed subspace.
By the induction hypothesis applied to $\bact_1$, $\nec{\bact_1}(\nec{\bact_2}\X)$ is a closed subspace.
Since $\nec{\bact_1\comp \bact_2}\X= \nec{\bact_1}(\nec{\bact_2}\X)$, 
$\nec{\bact_1\comp\bact_2}\X$ is a closed subspace.
\item[$\bact_1\cup\bact_2$]
By induction hypothesis, $\nec{\bact_1}\X$ and $\nec{\bact_2}\X$ are closed subspaces.
It follows that $\nec{\bact_1}\X \cap \nec{\bact_2}\X$ is a closed subspace.
Since $\nec{\bact_1\cup\bact_2}\X= \nec{\bact_1}\X \cap \nec{\bact_2}\X$,
we have that $\nec{\bact_1\cup\bact_2}\X$ is a closed subspace.

\item[$\bact^*$]
By induction hypothesis,
$\nec{\bact^n}\X$ is a closed subspace for all $n\in \N$.
\footnote{We let $\nec{\bact^0} \X= \X$.} 
It follows that $\bigcap_{n\in\N} \nec{\bact^n}\X$ is a closed subspace. 
Since $\nec{\bact^*}\X= \bigcap_{n\in\N} \nec{\bact^n}\X$, 
we have that $\nec{\bact^*}\X$ is a closed subspace.
\end{proofcases}
By induction hypothesis,
$\rho^{(W,M)}$ is a closed subspace.
By the statement above,  $\nec{\bact} \rho^{(W,M)}$ is a closed subspace.
\item[$\Forall{X}\rho_1$]
By induction hypothesis, $\rho_1^{(W_1,M)}$ is a closed subspace, for all $\Sigma[X]$-expansions $W_1$ of $W$.
Since the intersection of closed spaces is a closed subspace, $\displaystyle\bigcap_{W_1\red_\Sigma=W} \rho_1^{(W_1,M)}=(\Forall{X}\rho_1)^{(W,M)}$ is a closed subspace.
\end{proofcases}
\psqed\end{proof}

Closed sentences have some unique features which distinguish them from the rest of the sentences.
An immediate corollary of Theorem~\ref{th:Hibert-sat} is the following result.
\begin{corollary} \label{cor:Hilbert-sat}
For all models $(W,M)$ and all closed sentences $\rho$ over the same signature, we have:
\begin{enumerate}
\item $(W,M)\models^{(0^W)} \rho$, where $0:\to \vc$ stands for the origin vector.

\item If $(W,M)\models^w \rho$ then $(W,M)\models^{a w} \rho$, 
for all vectors $w\in W_\vc$ and all scalars $a\in W_\cn$.

\item If $(W,M)\models^{w_1} \rho$ and $(W,M)\models^{w_2} \rho$ then $(W,M)\models^{w_1 + w_2} \rho$, 
for all vectors $w_1,w_2\in W_\vc$.

\item Let $\{w_n\}_{n\in \N}$ be a Cauchy sequence of vectors, and let $w$ be its limit.

If $(W,M)\models^{w_n} \rho$ for all $n\in \N$ then $(W,M)\models^{w} \rho$.
\end{enumerate}
\end{corollary}

Obviously, each closed sentence is, in particular, a sentence, that is, $\Sen_c(\Delta)\subseteq \Sen(\Delta)$.
The other quantum operations are introduced as abbreviations.
For example, the Sasaki hook $\rho_1\leadsto \rho_2$ is defined as $\sim(\rho_1 \wedge \sim(\rho_1\wedge\rho_2))$, and the quantum disjunction $\rho_1\oplus\rho_2$ is defined as $\sim(\sim\rho_1\wedge \sim\rho_2)$ for all closed sentences $\rho_1$ and $\rho_2$.
The following result shows that the Sasaki hook can be viewed as a quantum implication for closed sentences.

\begin{lemma} \label{lemma:hook}
For all quantum models $(W,M)$ and 
all closed sentences $\rho_1$ and $\rho_2$ over the same signature $\Delta$, we have:
\begin{enumerate}
\item $\rho_1^{(W,M)}\cap(\rho_1\leadsto\rho_2)^{(W,M)}\subseteq \rho_2^{(W,M)}$, and
\item $\rho_1^{(W,M)}\subseteq \rho_2^{(W,M)}$ iff $(W,M)\models \rho_1\leadsto \rho_2$.
\end{enumerate}
\end{lemma}
\begin{proof} \

\begin{enumerate}
\item Assume that $(W,M)\models^w \rho_1$ and $(W,M)\models^w \rho_1\leadsto \rho_2$.
We show that $(W,M)\models^w\rho_2$: 
\begin{proofsteps}{19em}

 $w\in \rho_1^{(W,M)} \cap  (\rho_1^{(W,M)} \cap ((\rho_1 \wedge \rho_2)^{(W,M)})^\perp)^\perp$ & 
since  $(W,M)\models^w \rho_1$ and $(W,M)\models^w \rho_1\leadsto \rho_2$ \\

$\rho_1^{(W,M)} \cap (\rho_1^{(W,M)} \cap ((\rho_1 \wedge \rho_2)^{(W,M)})^\perp)^\perp=$ \newline
$((\rho_1\wedge \rho_2)^{(W,M)})^{\perp\perp} =(\rho_1\wedge \rho_2)^{(W,M)}$ &
by Theorem~\ref{th:hilbert}~(\ref{th:hilbert-3}), which says that the local closure is equal to the global closure \\

$w\in (\rho_1\wedge \rho_2)^{(W,M)}$ &
from the first two items above \\

$(W,M)\models^w\rho_1\wedge \rho_2$ & by semantics \\

$(W,M)\models^w\rho_2$ & by semantics
 \end{proofsteps}

\item We treat each implication separately.
\begin{proofcases}
\item[$\Rightarrow$]
Since $\rho_1^{(W,M)}\subseteq \rho_2^{(W,M)}$,
we have $(\rho_1\wedge \rho_2)^{(W,M)}=\rho_1^{(W,M)} \cap \rho_2^{(W,M)}= \rho_1^{(W,M)}$.
It follows that $(\rho_1\leadsto \rho_2)^{(W,M)} =(\sim(\rho_1\wedge\sim(\rho_1\wedge\rho_2)))^{(W,M)}=
(\sim(\rho_1\wedge \sim\rho_1))^{(W,M)}=W_\vc$.
\item[$\Leftarrow$]
Assume that $(W,M)\models\rho_1\leadsto \rho_2$. 
Let $w\in \rho_1^{(W,M)}$.
We have that $(W,M)\models^w\rho_1$ and $(W,M)\models^w\rho_1\leadsto \rho_2$.
By the proof above, $(W,M)\models^w\rho_2$ meaning that $w\in\rho_2^{(W,M)}$.
Since $w\in W_\vc$ was arbitrarily chosen, we get $\rho_1^{(W,M)}\subseteq \rho_2^{(W,M)}$.
\end{proofcases}
\end{enumerate}
\psqed\end{proof}

Applications to the specification and verification of quantum programs show that it is not necessary to interpret a formula as a closed subspace. See Section~\ref{sec:example} for the specification of two benchmark quantum protocols.
However, the present work includes results on closed sentences, which in turn, enable comparisons and applications of the present ideas to other modal variants of quantum logics.

\section{Initiality}
In this section, we define a notion of Horn clause in $\HDQL$ and then we prove the existence of initial models for any set of clauses. 
\begin{enumerate}
\item Let $\Sen_b(\Delta)$ be the set of sentences given by the grammar
\begin{center}
$\varphi \Coloneqq
p \mid
\varphi \wedge \varphi \mid
\at{k}\varphi \mid
\nec{\act}\varphi\mid
\store{z}\varphi_1$
\end{center} 
where
$p\in\Prop$ is any propositional symbol,
$k\in T_{\Sigma,\vc}$ is a term of sort vector, 
$\act$ is an action, 
$z$ is a variable for $\Delta$, and
$\varphi_1$ is a basic sentence over $\Delta[z]$.
We call these sentences \emph{locally basic}, or, simply, \emph{basic}~\cite{DBLP:journals/jacm/Gaina20}.
\item Let $\Sen_h(\Delta)$ be the set of Horn clauses given by the following grammar:
\begin{center}
$\gamma \Coloneqq
p \mid
\sim \varphi \mid
\gamma\wedge\gamma\mid
\rho_1\leadsto\rho_2\mid
\varphi\Rightarrow \gamma\mid
\at{k}\gamma\mid
\nec{\act}\gamma \mid 
\store{z}\gamma_1\mid
\Forall{X}\gamma_2$,
\end{center}
where
$p\in\Prop$ is any propositional symbol,
$\varphi$ is a basic sentence,
$\rho_1$ is a closed basic sentence,
$\rho_2$ is a closed Horn clause,
$k\in T_{\Sigma,\vc}$ is a term of sort vector,
$\act$ is an action,
$z$ is a vector variable for $\Delta$, 
$\gamma_1$ is a Horn clause over $\Delta[z]$,
$X$ is a set of vector variables for $\Delta$, and 
$\gamma_2$ is a Horn clause over $\Delta[X]$.
\end{enumerate}
In this this section, 
we fix a signature $\Delta=(\Sigma,\emptyset,\Prop)$ and 
a first-order model $\W$ over $\Sigma$ such that 
\begin{itemize}
\item $\W\red_{\Sigma^\hil}$ is a Hilbert space,
\item $u^\W:\W_\vc\to\W_\vc$ is a unitary transformation for all $u:\vc\to\vc\in U$, and
\item $q^\W:\W_\vc\to\W_\vc$ is a quantum measurement for all $q:\vc\to\vc\in Q$.
\end{itemize}
For the sake of simplicity, we assume that the sets of constants $C$ and $D$ are empty, which means that $\Sigma$ is of the form $(S^\hil,F^\hil\cup U\cup Q,P^\hil)$, where $(S^\hil,F^\hil,P^\hil)$ is the signature of Hilbert spaces, $U$ is a set of unitary transformation symbols and $Q$ is a set of measurement symbols.
We make the following notational conventions:
\begin{itemize} 
\item Let $\Sigma_\W$ be the first-order signature obtained from $\Sigma$ by adding all elements in $\W$ as constants, that is:
\begin{enumerate*}[label=(\alph*)]
\item $\Sigma_\W=(S^\hil,F^\hil\cup U \cup Q\cup D_\W \cup C_\W, P^\hil)$,
\item $D_\W=\W_\vc$, the set of vectors in $\W$, and 
\item $C_\W=\mathbb{C}$, the set of complex numbers.
\end{enumerate*}
\item Let $\W_\W$ be the first-order model over $\Sigma_\W$ obtained from $\W$ by interpreting each constant $c\in C_\W$ as the complex number $c$ and each constant $w\in D_\W$ as the vector $w$.
\item Let $E_\W$ be the set of first-order (ground) equations and relations satisfied by $\W_\W$, which means that $(\Sigma_\W,E_\W)$ is the positive diagram of $\W$.
\item We will work over the signature $\Delta_\W=(\Sigma_\W,E_\W,\Prop)$ for the rest of the paper.
\end{itemize}
In classical model theory, it is well-known that the $\Sigma_\W$-models which satisfy $E_\W$ are in one-to-one correspondence with the $\Sigma$-homomorphisms with the domain $\W$, that is, there is an isomorphism of categories $\Mod(\Sigma_\W,E_\W)\cong \W/\Mod(\Sigma)$.
The models of interest are described in the following definition.
\begin{definition} [Quantum models generated by sets of sentences]\label{def:init}
Any set of sentences $\Gamma\subseteq \Sen(\Delta_\W)$ defines a quantum model $(W^\Gamma,M^\Gamma)$ as follows:
\begin{enumerate*}[label=(\alph*)]
\item $W^\Gamma = \W_\W$, and 
\item $M^\Gamma:W^\Gamma_\vc\to |\Mod^\PL(\Sigma)|$ is the mapping defined by
$M^\Gamma_w=\{p\in\Prop\mid \Gamma\models^w p\}$ for all vectors $w\in W^\Gamma_\vc$.
\end{enumerate*}
\end{definition}
We show that if we restrict the set of sentences $\Gamma$ to \emph{Horn clauses} that are satisfiable, $(W^\Gamma,M^\Gamma)$ is the initial model of $\Gamma$. 

\subsection{Examples} \label{sec:example}
There are many examples of quantum gates, used in quantum circuits.
However,
the single qubit gates and
the controlled-NOT gate in the $2$-qubit system,
$$CNOT=\begin{pmatrix}
1 & 0 & 0 & 0\\
0 & 1 & 0 & 0\\
0 & 0 & 0 & 1\\
0 & 0 & 1 & 0
\end{pmatrix}$$ 
are the prototypes for all other gates.
Examples of frequently used single qubit gates are 
the Hadamard gate 
$$H=\displaystyle
\frac{1}{\sqrt{2}}
\begin{pmatrix}
1 & 1 \\
1 & -1
\end{pmatrix}$$
and the Pauli gates
$$X=
\begin{pmatrix}
0 & 1 \\
1 & 0
\end{pmatrix},
Y=
\begin{pmatrix}
0 & -i \\
i & 0
\end{pmatrix},
\text{ and }
Z=
\begin{pmatrix}
1 & 0 \\
0 & -1
\end{pmatrix}.$$
Therefore, fixing a Hilbert space together with some quantum gates and projective measurements, that is, a first-order model $\W$, is not a shortcoming from the point of view of applications. 
The Hilbert space together with the quantum gates and the projective measurements can be regarded as predefined data types necessary to specify the dynamics of quantum programs.
\paragraph{Superdense coding}
The superdense coding is a technique of sending two classical bits of information over a quantum channel.
Therefore, superdense coding involves two parties, traditionally known as Alice and Bob, which share a pair of qubits in the entangled state $\ket{\beta_{00}}=\displaystyle\frac{\ket{00}+\ket{11}}{\sqrt{2}}$.
Alice is initially in possession of the first qubit, while Bob has possession of the second qubit.
If Alice wishes to send the bits $ij$ to Bob, where $i,j\in\{0,1\}$, she applies the gate $X^i$ followed by the gate $Z^j$ to her qubit.
Then Bob is in possession of one of the four Bell states, 
\begin{center}
$\ket{\beta_{00}} = \displaystyle\frac{\ket{00}+\ket{11}}{\sqrt{2}}$,
$\ket{\beta_{01}} = \displaystyle\frac{\ket{01}+\ket{10}}{\sqrt{2}}$,
$\ket{\beta_{10}} = \displaystyle\frac{\ket{00}+\ket{11}}{\sqrt{2}}$,
$\ket{\beta_{11}} = \displaystyle\frac{\ket{01}-\ket{10}}{\sqrt{2}}$,
\end{center}
depending on the classical bits Alice wished to send to him.
Bob can now simply perform a measurement of the joint $2$-qubit state with respect to the Bell basis $\{\ket{\beta_{00}},\ket{\beta_{01}},\ket{\beta_{10}},\ket{\beta_{11}}\}$.

\begin{center}
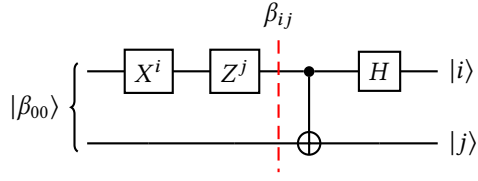
\begin{figure}
\begin{quantikz}
\lstick[wires=2]{\ket{\beta_{00}}} & \gate{X^i} & \gate{Z^j} \slice{$\beta_{ij}$} &  \ctrl{1} & \gate{H} & \qw \rstick{\ket{i}} \\
 & \qw & \qw  & \targ{} &  \qw &  \qw \rstick{\ket{j}} 
\end{quantikz}
\caption{Superdense coding}
\label{fig:superdense}
\end{figure}
\end{center}

The quantum circuit for superdense coding is depicted in Figure~\ref{fig:superdense}.
The measurement with respect to the Bell basis is depicted on the left side of the dashed vertical line in Figure~\ref{fig:superdense}.
If Bob receives $\ket{ij}$, this means that Alice sent him the bits $i$ and $j$, where $i,j\in\{0,1\}$.

The Hilbert space for this protocol is the $2$-qubit system $\Hil\otimes \Hil$, where $\Hil$ denotes the 2-dimensional Hilbert space.
The starting signature is $\Delta=(\Sigma,\emptyset,\Prop)$, where 
\begin{enumerate}
\item $\Sigma$ is obtained from the signature of Hilbert spaces $\Sigma^\hil$ defined in Section~\ref{sec:hilbert} by adding the set of unitary transformation symbols $U=\{ u_0,u_1,\sigma_0,\sigma_1,\delta_0,\delta_1\}$, and
\item $\Prop=\{p_{00},p_{01},p_{10},p_{11}\}$.
\end{enumerate}

The first-order model $\W$ over $\Sigma$ is obtained from the 2-qubit system $\Hil\otimes \Hil$ by interpreting the unitary transformation symbols as follows:
\begin{enumerate*}[label=(\alph*)]
\item $u_0^\W=CNOT$,
\item $u_1^\W=H\otimes I_2$,
\item $\sigma_0^\W=X^0\otimes I_2$,
\item $\sigma_1^\W=X^1\otimes I_2$, 
\item $\delta_0^\W=Z^0\otimes I_2$, and
\item $\delta_1^\W=Z^1\otimes I_2$,
\end{enumerate*}
where $I_n$ denotes the $n$-dimensional identity matrix for all $n>1$.
Notice that $X^0=Z^0=I_2$ and $X^0\otimes I_2=Z^0\otimes I_2=I_4$.

We define $\Gamma=\{\at{\ket{ij}}p_{ij} \mid  i,j \leq 1\}$.
The correctness of superdense coding means that $(W^\Gamma,M^\Gamma)\models^{\ket{\beta_{00}}} \nec{\act_{ij}} p_{ij}$, where $\act_{ij}=\sigma_i\comp\delta_j\comp u_0\comp u_1$ for all $i,j\in\{0,1\}$.
The results presented in Section~\ref{sec:basic} show that:
\begin{enumerate}
\item $(W^\Gamma,M^\Gamma)$ is the initial model of $\Gamma$, and 
\item $(W^\Gamma,M^\Gamma)\models^{\ket{\beta_{00}}} \nec{\act_{ij}} p_{ij}$ iff
$\Gamma\models^{\ket{\beta_{00}}} \nec{\act_{ij}} p_{ij}$ for all $i,j\in\{0,1\}$.
\end{enumerate}
If $(W^\Gamma,M^\Gamma)$ is the initial model of $\Gamma$ and $\Gamma\models^{\ket{\beta_{00}}} \nec{\act_{ij}} p_{ij}$ then 
$(W^\Gamma,M^\Gamma)\models^{\ket{\beta_{00}}} \nec{\act_{ij}} p_{ij}$. 
The proof of $\Gamma\models^{\ket{\beta_{00}}} \nec{\act_{ij}} p_{ij}$ consists of the following steps:

\begin{proofsteps}{25em}
$\Gamma\models^{\ket{\beta_{00}}} \nec{\sigma_i\comp\delta_j\comp u_0\comp u_1} p_{ij}$ &
by semantics\\
$\Gamma\models^{\ket{\beta_{00}}} \nec{\sigma_i}\nec{\delta_j}\nec{u_0}\nec{u_1} p_{ij}$ & 
by semantics\\
$\Gamma\models^{k_{ij}} p_{ij}$, where $k_{ij}=u_1(u_0(\delta_j(\sigma_i(\ket{\beta_{00}}))))$ is a term in $T_{\Sigma,\vc}$ & 
by semantics\\
$\Gamma\models^{\ket{ij}} p_{ij}$ &  
since $k_{ij}^\W=u_1^\W(u_0^\W(\ket{\beta_{ij}}))=\ket{ij}$\\
$\Gamma\models\at{\ket{ij}} p_{ij}$ which holds &
since $\{\at{\ket{ij}} p_{ij}\}\in\Gamma$
\end{proofsteps}

In the future, the arguments based on semantics used in the proof above will be replaced by arguments based on proof rules.
However, formal verification is out of the scope of the present paper.

\paragraph{Quantum teleportation}
Quantum teleportation is a technique for moving the state of a quantum system in the absence of a quantum communication channel linking the sender to the recipient.
There are two agents, called Alice and Bob,  who are separated in space. 
Each has one qubit of $\ket{\beta_{00}}=\displaystyle\frac{\ket{00}+\ket{11}}{\sqrt{2}}$.
In addition to her part of $\beta_{00}$, Alice holds a qubit $\w=\alpha \ket{0} +\beta\ket{1}$, where $\alpha$ and $\beta$ are unknown amplitudes.
Alice ``teleports'' the qubit $\ket{\w}$ to Bob, that is, she performs a program that has the input $\ket{\w_0}=\ket{\w}\otimes \ket{\beta_{00}}$ and the output $\ket{i}\otimes\ket{j} \otimes \ket{\w}$, where $i,j\in\{0,1\}$.
The quantum circuit is depicted in Figure~\ref{fig:teleport}.

The Hilbert space for this protocol is the $3$-qubit system $\Hil\otimes \Hil\otimes \Hil$, where $\Hil$ denotes the 2-dimensional Hilbert space.
The starting signature is $\Delta=(\Sigma,\emptyset,\Prop)$, where 
\begin{enumerate}
\item $\Sigma$ is obtained from the signature of Hilbert spaces $\Sigma^\hil$ defined in Section~\ref{sec:hilbert} by adding 
the set of unitary transformation symbols $U=\{ u_0,u_1,\sigma_0,\sigma_1,\delta_0,\delta_1\}$ and 
the set of measurement symbols $Q=\{q_{00},q_{01},q_{10},q_{11}\}$, and
\item $\Prop=\{p\}$.
\end{enumerate}

The first-order model $\W$ over $\Sigma$ is obtained from the 3-qubit system $\Hil\otimes \Hil\otimes \Hil$ by interpreting
\begin{enumerate}
\item the unitary transformation symbols as follows:
\begin{enumerate*}
\item $u_0^\W= CNOT\otimes I_2$, 
\item $u_1^\W=H\otimes I_4$,
\item $\sigma_0^\W=I_4\otimes X^0$, 
\item $\sigma_1^\W=I_4 \otimes X^1$,
\item $\delta_0^\W=I_4\otimes Z^0$, 
\item $\delta_1^\W=I_4\otimes Z^1$, and
\end{enumerate*}

\item the measurement symbols as follows: 
$q_{ij}^\W$ is the measurement corresponding to the projection on the closed subspace generated by $\{\ket{ij}\otimes\ket{0},\ket{ij}\otimes\ket{1}\}$, where $i,j\in\{0,1\}$.
\end{enumerate}

\begin{center}
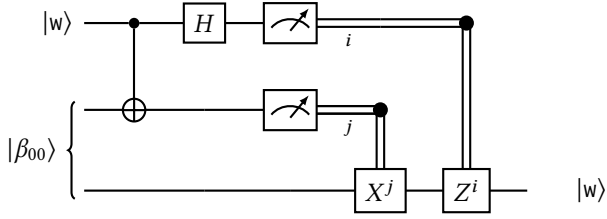
\begin{figure}
\begin{quantikz}
\lstick{\ket{\w}} & \ctrl{1} & \gate{H}  & \meter{} & \cw{i}  & \cwbend{2} \\
\lstick[wires=2]{$\ket{\beta_{00}}$} & \targ{} & \qw & \meter{} & \cwbend{1} \cw{j}   \\
  &  \qw  &   \qw   & \qw  & \gate{X^j} & \gate{Z^i} & \qw & \rstick{$\ket{\w}$}
\end{quantikz}
\caption{Quantum teleportation}
\label{fig:teleport}
\end{figure}
\end{center}

We define $\Gamma=\{\at{\ket{ij}\otimes \ket{\w}} p \mid i,j\leq 1 \}$.
The requirement is that the output is $\ket{i}\otimes\ket{j}\otimes \ket{\w}$, that is, 
$(W^\Gamma,M^\Gamma)\models^{\ket{\w}\otimes \ket{\beta_{00}}}  \nec{ \bigcup_{i,j\in\{0,1\}} \act_{ij}} p$, 
where $\act_{ij}=u_0\comp u_1 \comp q_{ij} \comp \sigma_j \comp \delta_i$ for all $i,j\in\{0,1\}$.
The results presented in Section~\ref{sec:basic} show that:
\begin{enumerate}
\item $(W^\Gamma,M^\Gamma)$ is the initial model of $\Gamma$, and
\item $(W^\Gamma,M^\Gamma)\models^{\ket{\w}\otimes \ket{\beta_{00}}}  \nec{ \bigcup_{i,j\in\{0,1\}} \act_{ij}} p$ iff 
$\Gamma\models^{\ket{\w}\otimes \ket{\beta_{00}}} \nec{\bigcup_{i,j\in\{0,1\}} \act_{ij}} p$.
\end{enumerate}
The proof of $\Gamma\models^{\ket{\w}\otimes \ket{\beta_{00}}} \nec{\bigcup_{i,j\in\{0,1\}} \act_{ij}} p$ consists of the following steps:
\begin{proofsteps}{27em}
$\Gamma\models^{\ket{\w}\otimes \ket{\beta_{00}}} \nec{\bigcup_{i,j\in\{0,1\}} \act_{ij}} p$ & 
by semantics\\

$\Gamma\models^{\ket{\w}\otimes \ket{\beta_{00}}} \nec{\act_{ij}} p$ for all $i,j\in\{0,1\}$ &
by semantics\\
$\Gamma\models^{\ket{\w}\otimes \ket{\beta_{00}}}  \nec{u_0\comp u_1 \comp q_{ij} \comp \sigma_j \comp \delta_i} p$ for all $i,j\in\{0,1\}$ & 
by semantics\\
$\Gamma\models^{k_{ij}} p$ for all $i,j\in\{0,1\}$, \newline 
where $k_{ij}=\delta_i(\sigma_j(q_{ij}(u_1(u_0(\ket{\w}\otimes \ket{\beta_{00}})))))$ for all $i,j\in\{0,1\}$  & 
by semantics\\
$\Gamma\models^{\ket{ij}\otimes\ket{w}}p$ for all $i,j\in\{0,1\}$ & 
since $k_{ij}^\W=\ket{ij}\otimes\ket{w}$ \\
$\Gamma\models \at{\ket{ij}\otimes\ket{w}}p$ for all $i,j\in\{0,1\}$, which holds & 
since $\{\at{\ket{ij}\otimes\ket{w}}p\}\in\Gamma$
\end{proofsteps}

\subsection{Basic level} \label{sec:basic}
In the fragment obtained by restricting the syntax to basic sentences,
the validity of the semantic entailment relation can be established by considering only the initial models described in Definition~\ref{def:init}.
In other words, the satisfaction of basic sentences is encapsulated in initial models.
\begin{theorem} \label{th:basic}
Let $\Gamma\subseteq \Sen(\Delta_\W)$ be any set of sentences (not necessarily basic).
Then $\Gamma\models^w\varphi$ iff $( W^\Gamma,M^\Gamma)\models^w\varphi$ 
for all basic sentences $\varphi\in\Sen_b^\HDQL(\Delta)$ and all vectors $w\in \W_\vc$.
\end{theorem}
\begin{proof}
We proceed by induction on the structure of $\varphi$:
\begin{proofcases}
\item[$p\in\Prop$]  Straightforward, by the definition of $( W^\Gamma,M^\Gamma)$.
\item[$\varphi_1\wedge\varphi_2$] 
$\Gamma\models^w \varphi_1\wedge \varphi_2$ iff
$\Gamma\models^w \varphi_1$ and $\Gamma\models^w \varphi_2$ iff (by induction hypothesis)
$( W^\Gamma,M^\Gamma)\models^w\varphi_1$ and $( W^\Gamma,M^\Gamma)\models^w\varphi_2$ iff
$( W^\Gamma,M^\Gamma)\models^w\varphi_1\wedge \varphi_2$.
\item [$\at{k}\varphi$]
$\Gamma\models^w\at{k}\varphi$ iff  
$\Gamma\models^v\varphi$, where $v=k^\W$ iff (by induction hypothesis)
$(W^\Gamma,M^\Gamma)\models^v\varphi$ iff (by semantics)
$(W^\Gamma,M^\Gamma)\models^w\at{k}\varphi$.
\item [$\nec{\act}\varphi$] 
We proceed by induction on the structure of the actions $\act$:
\begin{proofcases}
\item [$u\in U$]
$\Gamma\models^w \nec{u}\varphi$ iff
$\Gamma\models^v\varphi$, where $v=u^\W(w)$ iff (by induction hypothesis)
$( W^\Gamma,M^\Gamma)\models^v\varphi$ iff
$( W^\Gamma,M^\Gamma)\models^w\nec{u}\varphi$.
\item [$q\in Q$] This case is similar to the one above.
\item [$\act_1\comp\act_2$]
By the induction hypothesis applied to $\act_2$, we have
$\Gamma\models^w \nec{\act_2}\varphi$ iff
$( W^\Gamma,M^\Gamma)\models^w \nec{\act_2}\varphi$.
By the induction hypothesis applied to $\act_1$, we obtain
$\Gamma\models^w \nec{\act_1}\nec{\act_2}\varphi$ iff
$( W^\Gamma,M^\Gamma)\models^w \nec{\act_1}\nec{\act_2}\varphi$.
Hence, 
$\Gamma\models^w \nec{\act_1\comp\act_2}\varphi$ iff 
$( W^\Gamma,M^\Gamma)\models^w\nec{\act_1\comp\act_2}\varphi$.
\item [$\act_1\cup\act_2$]
$\Gamma\models^w \nec{\act_1\cup\act_2}\varphi$ iff 
$\Gamma\models^w \nec{\act_1}\varphi$ and 
$\Gamma\models^w \nec{\act_2}\varphi$ iff (by induction hypothesis)
$( W^\Gamma,M^\Gamma)\models^w\nec{\act_1}\varphi$ and
$( W^\Gamma,M^\Gamma)\models^w\nec{\act_2}\varphi$ iff
$( W^\Gamma,M^\Gamma)\models^w\nec{\act_1\cup \act_2}\varphi$.
\item [$\act^*$]
$\Gamma\models^w \nec{\act^*}\varphi$ iff 
$\Gamma\models^w \nec{\act^n}\varphi$ for all $n\in \N$ iff (by the induction hypothesis)
$( W^\Gamma,M^\Gamma)\models^w\nec{\act^n}\varphi$ for all $n\in\N$ iff 
$( W^\Gamma,M^\Gamma)\models^w\nec{\act^*}\varphi$.
\end{proofcases}
\item[$\store{z}\varphi_1$] 
$\Gamma\models^w\store{z}\varphi_1$ iff $\Gamma\models^w\theta(\varphi_1)$, where $\theta:\{z\}\to T_{\Sigma,\vc}$ is defined by $\theta(z)=w$ iff (by induction hypothesis) 
$(W^\Gamma,M^\Gamma)\models^w\theta(\varphi_1)$ iff (by semantics) 
$(W^\Gamma,M^\Gamma)\models^w\store{z}\varphi_1$.
\end{proofcases}
\psqed\end{proof}
Theorem~\ref{th:basic} holds even if $\Gamma$ is not a set of basic sentences.
\begin{corollary}\label{cor:basic}
If $\Gamma$ is a set of basic sentences then $( W^\Gamma,M^{\Gamma})$ is the initial model of $\Gamma$.
\end{corollary}
\begin{proof}
Since $\Gamma\subseteq \Sen_b(\Delta_\W)$, by Theorem~\ref{th:basic}, $(W^\Gamma,M^\Gamma)\models\Gamma$.
For any quantum model $(V,N)$ which satisfies $\Gamma$, since $W^\Gamma=\W_\W$ is the initial first-order model of the theory $(\Sigma_\W,E_\W)$, there exists a unique first-order homomorphism $h:W^\Gamma\to V$.
By Theorem~\ref{th:basic}, $p\in M^\Gamma_w$ implies $p\in N_{h(w)}$, for all propositional symbols $p\in \Prop$ and all vectors $w\in W^\Gamma_\vc$. 
Therefore, $h:W^\Gamma\to V$ defines a homomorphism $h:(W^\Gamma,M^\Gamma)\to (V,N)$ of quantum models.
Since $h:W^\Gamma\to V$ is unique, $h:(W^\Gamma,M^\Gamma)\to (V,N)$ is unique too.
\end{proof}

\section{Upper level}
There are sets of Horn clauses which are not satisfiable. 
Assume that $\W$ has at least one vector $w\in \W_\vc$ different from $0^\W$.
Then take, for example, $\Gamma=\{p,\sim p\}$, where $p$ is any propositional symbol.
Suppose towards a contradiction that $(V,M)\models\Gamma$ for some quantum model $(V,M)$ over $\Delta_\W=(\Sigma_W,E_\W,\Prop)$.
Since $V\models^\FOL E_\W$, by Lemma~\ref{lemma:inj}, there exists an injective homomorphism $h:W^\Gamma\to V$, which implies that $h(w)\neq 0^V$.
Since  $(V,M)\models\Gamma$, we have $(V,M)\models^{h(w)} p$ and $(V,M)\models^{h(w)} \sim p$, which implies $h(w)\in p^{(V,M)} \cap (p^{(V,M)})^\perp$.
Since $p^{(V,M)} \cap (p^{(V,M)})^\perp=0^V$, we obtain $h(w)=0^V$, which is a contradiction.
Therefore, special attention is needed when using quantum negation.
Despite this fact, we show that any consistent theory in $\Sen_h(\Delta_\W)$ has an initial model.
\begin{theorem} \label{th:init} 
Let $\Gamma\subseteq \Sen(\Delta_\W)$ be any satisfiable set of sentences.
Then: 
\begin{center}
$( W^\Gamma,M^\Gamma)\models^w \gamma$ if $\Gamma\models^w\gamma$ 
\end{center}
for all Horn clauses $\gamma\in\Sen_h(\Delta_\W)$ and all vectors $w\in\W_\vc$.
\end{theorem}

\begin{proof}
We proceed by induction on the structure of $\gamma$.
The base case holds by Theorem~\ref{th:basic}.
The cases corresponding to retrieve $@$, store $\downarrow$, necessity $\nec{\act}$ can be proved using ideas from the proof of Theorem~\ref{th:basic}.
We focus on the remaining cases.
\begin{proofcases}
\item[$\sim\varphi$] 
Assume that $\Gamma\models^w\sim\varphi$. 
For all $v\in \varphi^{( W^\Gamma,M^\Gamma)}$ we have:
\begin{proofsteps}{17em}
$\Gamma\models^v\varphi$ & by Theorem~\ref{th:basic}, since $( W^\Gamma,M^\Gamma)\models^v \varphi$\\
$(V,M)\models \Gamma $ for some model $(V,M)$ & since $\Gamma$ is satisfiable\\
$V\models^\FOL  \ip{v}{w}=0 $ & since $\Gamma\models^w\sim\varphi$ and $\Gamma\models^v\varphi$ and $(V,M)\models \Gamma$  \\
there exists $h:W^\Gamma\to V$ injective  & 
by Lemma~\ref{lemma:inj}, since $V\models^\FOL E_\W$\\
$W^\Gamma\models^\FOL \ip{v}{w}=0$ & 
since $V\models^\FOL  \ip{v}{w}=0$ and $h$ is injective
\end{proofsteps}
It follows that $w\in (\varphi^{(W^\Gamma,M^\Gamma)})^\perp$, which means $( W^\Gamma,M^\Gamma)\models^w \sim\varphi$.
\item [$\rho_1\leadsto \rho_2$]
By Lemma~\ref{lemma:hook}~(2), we need to show that $\rho_1^{(W^\Gamma,M^\Gamma)}\subseteq\rho_2^{(W^\Gamma,M^\Gamma)}$.
Assume $w\in \rho_1^{(W^\Gamma,M^\Gamma)}$, which means $(W^\Gamma,M^\Gamma)\models^w \rho_1$.
Since $\rho_1\in\Sen_b(\Delta_\W)$, by Theorem~\ref{th:basic}, we have $\Gamma\models^w\rho_1$.
Since $\Gamma\models^w\rho_1\leadsto \rho_2$, by Lemma~\ref{lemma:hook}~(1), we have $\Gamma\models^w\rho_2$.
Since $\rho_2$ is a Horn clause, by induction hypothesis, $(W^\Gamma,M^\Gamma)\models^w\rho_2$.
Hence, $w\in \rho_2^{(W^\Gamma,M^\Gamma)}$.
\item[$\varphi\Rightarrow\gamma$]
Assume that $\Gamma\models^w\varphi\Rightarrow \gamma$.
\begin{itemize}
\item [] If $(W^\Gamma,M^\Gamma)\models^w\varphi$ then since $\varphi\in\Sen_b(\Delta_\W)$, by Theorem~\ref{th:basic}, $\Gamma\models^w\varphi$. 
Since $\Gamma\models^w\varphi\Rightarrow \gamma$ and $\Gamma\models^w\varphi$, we get $\Gamma\models^w\gamma$.
By induction hypothesis, $(W^\Gamma,M^\Gamma)\models^w\gamma$.
\end{itemize}
Hence, $(W^\Gamma,M^\Gamma)\models^w\varphi\Rightarrow \gamma$.
\item[$\Forall{X}\gamma$] 
Assume that $\Gamma\models^w\Forall{X}\gamma$.
\begin{itemize}
\item[] Let $(V,M^\Gamma)$ be a $\Delta[X]$-expansion of $(W^\Gamma,M^\Gamma)$.
Let $\theta:X\to W^\Gamma$ be the substitution defined by $\theta(x)=x^V$ for all $x\in X$.
By the definition of $\theta$, $(W^\Gamma,M^\Gamma)\red_\theta=(V,M^\Gamma)$.
Since $\Gamma\models^w\Forall{X}\gamma$, we have $\Gamma\models^w\theta(\gamma)$.
By induction hypothesis, $(W^\Gamma,M^\Gamma)\models^w\theta(\gamma)$.
Since $(W^\Gamma,M^\Gamma)\red_\theta=(V,M^\Gamma)$, by the satisfaction condition for substitutions, $(V,M^\Gamma)\models^w \gamma$.
\end{itemize} 
Since $(V,M^\Gamma)$ is an arbitrary $\Delta_\W[X]$-expansion of $(W^\Gamma,M^\Gamma)$, 
we obtain $(W^\Gamma,M^\Gamma)\models^w\Forall{X}\gamma$.
\end{proofcases}
\end{proof}

If we remove quantum negation from the definition of Horn clauses, the satisfiability condition is not needed.
In the absence of quantum negation, all Horn clauses have an initial model as in first-order logic. 

\begin{corollary}\label{cor:init}
If $\Gamma$ is a satisfiable set of Horn clauses then $( W^\Gamma,M^{\Gamma})$ is the initial model of $\Gamma$.
\end{corollary}
A query is a sentence of the form $\Exists{X} \wedge E$, where $X$ is a set of vector variables and $E\subseteq \Sen_b(\Delta_\W[X])$ is a finite set of basic sentences over $\Delta_\W[X]$.
The following result states that the (semantic) satisfaction of a query by a program understood here as a set of Horn clauses is equivalent to finding a (syntactic) \emph{substitution} for that query.

\begin{theorem}[Herbrand's Theorem] \label{th:LP}
For all satisfiable sets of Horn clauses $\Gamma$, all queries $\Exists{X}\wedge E$ and all vectors $w\in W^\Gamma_\vc$, the following are equivalent:
\begin{enumerate}[leftmargin=*]
\item~\label{th:LP1}$\Gamma\models^w\Exists{X}\wedge E $,
\item~\label{th:LP2}$(W^\Gamma,M^\Gamma)\models^w\Exists{X}\wedge E$, and
\item~\label{th:LP3}$\Gamma \models^w \Forall{Y}\wedge \theta(E)$ for some substitution $\theta:X\to T_{(\Sigma_\W)}(Y)$,
\end{enumerate}
where $(W^\Gamma,M^\Gamma)$ is the model described in Definition~\ref{def:init}.
\end{theorem}
\begin{proof}
By Corollary~\ref{cor:init}, $(W^\Gamma,M^\Gamma)$ is the initial model of $\Gamma$. 
\begin{proofcases}
\item [\ref{th:LP1} $\Rightarrow$ \ref{th:LP2}]
By the fact that $(W^\Gamma,M^\Gamma)\models \Gamma$.
\item [\ref{th:LP2} $\Rightarrow$ \ref{th:LP3}] 
Assume that $(W^\Gamma,M^\Gamma)\models^w \Exists{X}\wedge E$. 
\begin{proofsteps}{22em}
$(V,M^\Gamma)\models^w E$ for some quantum model \newline $(V,M^\Gamma)\in|\Mod(\Delta_\W[X])|$ such that $V\red_{(\Sigma_\W)}=W^\Gamma$ & by semantics \\
let $\theta:X\to T_{\Sigma_\W}$ defined by \rlap{$\theta(x)=x^V$ for all $x\in X$ } & \\
$(W^\Gamma,M^\Gamma)\red_\theta=(V,M^\Gamma)$ & as $(W^\Gamma\red_\theta)_x=\theta(x)^{(W^\Gamma)}= x^V$ for all $x\in X$ \\
$(W^\Gamma,M^\Gamma)\models^w \theta(E)$ & since $(V,M^\Gamma)\models^w E$ \\
$\Gamma\models^w \theta(E)$ & by Theorem~\ref{th:basic} 
 \end{proofsteps}
\item [\ref{th:LP3} $\Rightarrow$ \ref{th:LP1}]
Assume that $\Gamma\models^w\Forall{Y}\theta(E)$ for some substitution $\theta:X\to T_{(\Sigma_W)}(Y)$.
Let $(W,M)$ be a model of $\Gamma$.
We show that $(W,M)\models^w\Exists{X} \wedge E$.
Let $(V,M)$ be any expansion of $(W,M)$ to the signature $\Delta_\W[Y]$.
We have that $(V,M)\models\Gamma$ and since $\Gamma\models^w\Forall{Y}\theta(E)$, 
we get $(V,M)\models^w\theta(E)$.
By the satisfaction condition for substitutions, $(V,M)\red_\theta\models^w E$. 
Since $(V,M)\red_\theta\red_{\Delta_\W}=(W,M)$, we have $(W,M)\models^w\Exists{X} \wedge E$.
Since $(W,M)$ is an arbitrary quantum model of $\Gamma$, 
we obtain $\Gamma\models^w\Exists{X} \wedge E$.
\end{proofcases}
\end{proof}
\section{Conclusions}
We have presented a new hybrid-dynamic logic for the specification and analysis of quantum systems.
From a semantic perspective, the present logical framework allows one to capture quantum systems as Kripke structures whose states 
\begin{enumerate*}[label=(\alph*)]
\item are the vectors of a Hilbert space, and
\item are labelled with propositional logic models that describe the evolution of systems.
\end{enumerate*}
We gave two benchmark examples of quantum protocols that can be captured in our logical framework: superdense coding and quantum teleportation.
We proposed a notion of Horn clause, and we proved two results of great importance which set the foundation of quantum logic programming: 
the existence of initial models of any satisfiable set of Horn clauses and a variant of Herbrand's Theorem.
The satisfiability requirement distinguishes the present results from previous approaches to initiality.
Due to quantum negation not all sets of Horn clauses are consistent. Initiality and consistency are conflicting concepts, in general. Please recall that in first-order logic, any set of Horn clauses is consistent. We proved that they can coexist in hybrid-dynamic quantum logic under the satisfiability requirement. This is the first paper which studies initiality in quantum computing and lays the foundation for automated reasoning (which belongs to the area of model theory).

The present work builds on the ideas used to set the foundations of logic programming in hybrid logics presented in \cite{DBLP:journals/tcs/Gaina17} and is related to the developments in the initial semantics for hybridized institutions reported in~\cite{Diaconescu16}. 
However, the hybrid-dynamic quantum logic defined in the present contribution has a number of features that take it outside the framework proposed in \cite{DBLP:journals/tcs/Gaina17,Diaconescu16}:
\begin{enumerate*}[label=(\alph*)]
\item The set of states is no longer plain, but it has a complex algebraic structure given by a Hilbert space.
\item Necessity is defined over structured actions constructed from projective measurements and unitary transformations not just plain binary relations.
\item Sentence building operators include quantum negation $\sim$ and quantum implication $\leadsto$.
\end{enumerate*}
Compared to~\cite{Diaconescu16}, the initiality results developed here are based on inductive arguments that avoid the heavier model-theoretic infrastructure of quasi-varieties and inclusion systems.
At this stage, it is not clear if the class of quantum models over a given signature is a quasi-variety but it is an interesting direction of research for the future.
The present contribution is also related to the developments in the foundations of logic programming at an abstract level from~\cite{Diaconescu04,TutuF15,TutuF17}. 
However, the category-based framework developed there is not equipped with features of modal logics.

The method of diagrams proposed by Robinson in model theory is used to describe Hilbert spaces.
In applications, this approach corresponds to developing a library of predefined data types for scalars and vectors and the operations on them.
Such a library of predefined data types is a necessary feature for a quantum programming language.
When writing a quantum program, the engineers focus on describing the dynamics of quantum systems by relying on such features not necessarily thinking of developing them further.
Similarly, we propose a quantum logic for specifying the dynamics of quantum systems by relying on the existing knowledge of linear algebra.
 
An important task is to study Birkhoff completeness and resolution-based inference rules for hybrid-dynamic quantum logic along the lines of \cite{DBLP:conf/tableaux/GainaT19,DBLP:journals/fac/Gaina17}.
Recent studies show that classical computing can aid unreliable quantum processors of small and intermediate size to solve large problems reliably \cite{osti_10084839}.
In the future, we are planning to equip our hybrid-dynamic quantum logic with features that support the description of hybrid programs.
We will study logical properties that involve all sentences not only Horn clauses following, for example, ideas used to prove omitting types theorem and interpolation from \cite{DBLP:journals/apal/GainaBK23} and \cite{DBLP:conf/aiml/BadiaKG22}, respectively. 

\bibliographystyle{acm}
\bibliography{init}
\end{document}